\documentclass[]{imsart}
\pubyear{2023}
\volume{TBA}
\issue{TBA}
\arxiv{arXiv:2209.00636}
\firstpage{1}
\lastpage{1}

\usepackage{amsthm}
\usepackage{amsmath}
\usepackage{amsfonts}
\usepackage{natbib}
\usepackage[colorlinks,citecolor=blue,urlcolor=blue,filecolor=blue,backref=page]{hyperref}
\usepackage{graphicx}

\startlocaldefs
\numberwithin{equation}{section}
\theoremstyle{plain}

\newtheorem{proposition}{Proposition}[section]
\endlocaldefs

\begin{document}

\begin{frontmatter}
\title{A Conservation Law for Posterior Predictive Variance}
\runtitle{On Posterior Predictive Variance}

\begin{aug}
\author{\fnms{Dean} \snm{Dustin}\thanksref{addr1}\ead[label=e1]{ddustin8@huskers.unl.edu}},
\author{\fnms{Bertrand} \snm{Clarke}\thanksref{addr2}\ead[label=e2]{bclarke3@unl.edu}}

\runauthor{Dustin and Clarke}

\address[addr1]{Charles Schwab Financial, Denver, CO
    \printead{e1} % print email address of "e1"
}

\address[addr2]{Department of Statistics, University of Nebraska-Lincoln, NE, USA, 68583-0963
  \printead{e2}
}

%\thankstext{t1}{Some comment}

\end{aug}

\begin{abstract}
We use the law of total variance to generate multiple expressions for the posterior predictive variance 
in Bayesian hierarchical models.   These expressions are sums of terms involving
conditional expectations and conditional variances.    Since the posterior predictive
variance is fixed given the hierarchical model, it represents a constant quantity that is
conserved over the various expressions for it.   The terms in the expressions can be
assessed in absolute or relative terms to understand the main contributors to the
length of prediction intervals.   Also, sometiems these terms can be intepreted
in the context of the hierarchical model.  We show several examples, closed form
and computational, to illustrate the uses of this approach in model assessment.
\end{abstract}

\begin{keyword}[class=MSC]
\kwd[Primary ]{62F15}
\kwd[; secondary ]{62J10}
\end{keyword}

\begin{keyword}
\kwd{prediction interval} \kwd{posterior predictive variance} \kwd{law of total variance}
\kwd{Bayes model averaging}
\end{keyword}

\end{frontmatter}

\section{The Setting and Intution}
\label{Intro}

Consider a generic Bayesian hierarchical model for a response $Y=y$ given 
$V = V_K= (V_1, \ldots , V_k, \ldots, V_K)^T$ taking values $v = (v_1, \ldots , v_K)^T$
for
some $K \in \mathbb{N}$:
\begin{eqnarray}
V_1 &\sim& w(v_1) \nonumber \\
V_2 &\sim& w(v_2 \vert v_1) \nonumber \\
\vdots & \vdots  &  \vdots  \nonumber \\
V_K &\sim& w(v_K \vert v_1, \ldots , v_{K-1}) \nonumber \\
Y &\sim& p(y \vert v),
\label{hierarchicalmodelgeneric}
\end{eqnarray}
where the $w$'s represent prior densities for the $V_k$'s as indicated by their
arguments and $p(\cdot \vert v)$ is the likelihood.   All densities are with respect
to Lebesgue measure when the random variable is continuous.  For discrete random variables
we regard the density as being taken with respect to counting measure.
We denote $n$ outcomes of $Y$ by $Y^n = (Y_1, \ldots , Y_n)^T$ with outcomes
$y^n = (y_1, \ldots , y_n)^T$.  

It is common practice to adopt an estimation perspective.  That is,
choose a parameter,  here one of the $V_k$'s, and obtain credibility sets for it
from the posterior $w(v_k \vert y^n)$.  If the credibility set
for a given $V_k$ is sufficiently small as determined by hypothesis testing, say, then it may make
sense to drop the $k$-th level of the hierarchy.  However, it is unclear in the abstract how to
compare the length of a credibility set for one $V_k$ to the length of a credibility set for
$V_{k^\prime}$ for $k \neq k^\prime$.  Aside from asymptotics usually based on the Fisher information,
there is no common scale on which the variances of different
$V_k$'s can be compared.   The reason is that the 
size of $\hbox{Var}(V_k \vert y^n)$ is unrelated to
the size of $\hbox{Var}(V_{k^\prime} \vert y^n)$.   Nothing necessarily 
ties the $K$ marginal posteriors $w(v_k \vert y^n)$ together with a common scale
p[re-asymptotically.    Indeed, when estimating
a value of $v_k$, it is not in general clear how the sizes of other $v_{k^\prime}$'s affect it.

An alternative analysis of hierarchical models follows from a predictive perspective.
Instead of looking at posterior variances, we look at terms that sum to the posterior
predictive variance and compare their relative importance.  Without further discussion,
we assume that posterior means and variances in general are the right quantities to study.  This is true under
squared error loss.  However, other choices of loss function would yield different, but analogous, reasoning.

Given $y^n$,  we assign the posterior predictive density to future values $Y_{n+1}$, that is 
\begin{eqnarray}
Y \sim
p(y_{n+1} \vert y^n) = \int p(y_{n+1} \vert v) w(v \vert y^n) {\rm d} v,
\label{postpredfut}
\end{eqnarray}
where $w(v \vert y^n)$ is the posterior density and ${\rm d}v$ is summation or integration
as appropriate.  At this point the posterior predictive variance within the context of
the model \eqref{hierarchicalmodelgeneric}
is fixed.  Denote it $\hbox{Var}(Y_{n+1} \vert y^n)$.   When a random variable in the top
$K$ levels of the hierarchy is visble we say it is explicit  Otherewise we say it is implicit.
Thus,  $\hbox{Var}(Y_{n+1} \vert y^n)$ depends implicitly on the top $K$ levels of \eqref{hierarchicalmodelgeneric}.

Recall the standard probability theory result called the Law of Total Variance (LTV).  Generically,
for random variables $W$ and $Z$ on the same probability space it is
\begin{eqnarray}
\hbox{Var}(W)  = E [ \hbox{Var} (W\vert Z) ]+ \hbox{Var} [E (W \vert Z)].
\label{LTV}
\end{eqnarray}
By reinterpreting \eqref{LTV} in the posterior context we have
\begin{eqnarray}
\hbox{Var}(W \vert y^n)  =  \hbox{Var}_{W\vert y^n} (W \vert y^n) = E_{Z \vert y^n} [ \hbox{Var} (W\vert Z, y^n) ]+ \hbox{Var}_{Z \vert y^n} [E (W \vert Z, y^n)],
\label{LTVP}
\end{eqnarray}
assuming that $W$ and $Z$ are functions on the same probability space as used to
write \eqref{hierarchicalmodelgeneric}.  In \eqref{LTVP}
the outer expectations on the right, usually suppresed,  are indicated and
both sides are functions of $y^n$.
 
The predictive approach takes $W$ to be a future value $Y_{n+1}$, rather than any of the $v_k$'s.
For generality, we will also take $y^n$ to be the pre-$n+1$ data, and hence condition
on ${\cal{D}} = {\cal{D}}_n$.    
In contrast to $y^n$,  ${\cal{D}}$ may include values of explanatory variables for each time step.
Unless stated otherwise, we assume the data is independent from time step to time step.
Now we have 
\begin{eqnarray}
\hbox{Var}_{Y_{n+1} \vert {\cal{D}} }( Y_{n+1} \vert {\cal{D}}) = E_{Z \vert {\cal{D}}} [ \hbox{Var} ( Y_{n+1}\vert Z, {\cal{D}}) ]+ \hbox{Var}_{Z \vert {\cal{D}}} [E (Y_{n+1} \vert Z, {\cal{D}})].
\label{postpredvar}
\end{eqnarray}
Independent of the choice of $Z$, the left hand side of \eqref{postpredvar} is a constant
depending only on the hierarchy \eqref{hierarchicalmodelgeneric} and ${\cal{D}}$.  That is,
\eqref{postpredvar} is a conservation law for the posterior predictive variance over choices
of conditioning.    We can choose $Z$ to be any function of a subset of the entries of $V$.
In particular, if $Z= V_1$, we get
\begin{eqnarray}
\hbox{Var}_{Y_{n+1} \vert {\cal{D}} }( Y_{n+1} \vert {\cal{D}}) = E_{V_1 \vert {\cal{D}}} [ \hbox{Var} ( Y_{n+1}\vert V_1, {\cal{D}}) ]+ \hbox{Var}_{V_1 \vert {\cal{D}}} [E (Y_{n+1} \vert V_1, {\cal{D}})].
\label{postpredvarV1}
\end{eqnarray}

More is true.   The LTV can be applied iteratively to the terms in \eqref{postpredvar}.
 It is seen that
$ \hbox{Var} ( Y_{n+1}\vert Z, {\cal{D}})$, the first term on the right in \eqref{postpredvar},
 is of the same form as the left hand side of
\eqref{postpredvar} -- simply replace ${\cal{D}}$ by $(Z, {\cal{D}})$.
If we take $Z = V_1$, condition on $V_1 = v_1$, and use another instance of the LTV, this time with $Z = V_2$, we get
\begin{eqnarray}
\hbox{Var} ( Y_{n+1}\vert v_1, {\cal{D}}) = 
E_{V_2 \vert v_1,{\cal{D}}} [ \hbox{Var} ( Y_{n+1}\vert V_1, V_2, {\cal{D}}) ]+ \hbox{Var}_{V_2 \vert v_1, {\cal{D}}} [E (Y_{n+1} \vert V_1, V_2, {\cal{D}})].
\label{LTVgeniter}
\end{eqnarray}
Using \eqref{LTVgeniter} in \eqref{postpredvar} we get, with some simplification of notation,
\begin{eqnarray}
\hbox{Var}( Y_{n+1} \vert {\cal{D}}) &=& 
E_{V_1 \vert {\cal{D}} } E_{V_2 \vert V_1,{\cal{D}}} [ \hbox{Var} ( Y_{n+1}\vert V_1, V_2, {\cal{D}}) ] \nonumber \\
&& +
E_{V_1 \vert {\cal{D}} } \hbox{Var}_{V_2 \vert V_1, {\cal{D}}} [E (Y_{n+1} \vert V_1, V_2, {\cal{D}})] \nonumber\\
&& +
\hbox{Var}_{V_1 \vert {\cal{D}}} [E (Y_{n+1} \vert V_1, {\cal{D}})].
\label{LTV2termgen}
\end{eqnarray}

Now, \eqref{postpredvarV1} and \eqref{LTV2termgen} are two expressions for the same
$\hbox{Var}( Y_{n+1} \vert {\cal{D}})$.    They are generic in that the role of $V_1$ and
$V_2$ can be played by any two functions of entries of $V$.  That is, the posterior predictive
variance admits a very large number of two term and three term generic expressions.

This procedure can be iterated in multiple ways to include any other $V_k$, thereby generating even more
expressions for $\hbox{Var}( Y_{n+1} \vert {\cal{D}})$.  Indeed, every time an
expression of the form $\hbox{Var}(Y_{n+1} \vert W, {\cal{D}})$ for any suitable random variable $W$ occurs from using the
LTV, the LTV can be applied again provided
a further suitable conditioning variable $Z$ can be found.  That is, 
the conservation law for posterior predictive variance in \eqref{postpredvarV1} extends to a far larger class of sums of terms
involving conditional expectations and variances than \eqref{postpredvar} initially suggests.

We call the collection of expressions for the posterior predictive variance in a fixed hierarchical model
its LTV-scope.  Thus, the posterior predictive variance is invariant or conserved over its scope.  We regard the
introduction of an extra level in a hierarhical model as creating a new model and hence a new LTV-scope.
The point of this work is not only to look within the LTV-scope of one hierarchical model but to compare
LTV-scopes across models.   Expressions in the scope of a hierarchical model also admit an interpretation
in terms of analysis of variance and associated frequentist testing, 
see \cite{Dustin:Clarke:2023}, but we do not discuss this here.
Fixing a hierarchical model and looking at the scope of its posterior predictive variance
lets us choose which decomposition has the interpretation we want to use to
decide which `components' are more important than
other components in relative or absolute terms.
Otherwise put, we can examine and compare multiple decompositions of the
posterior variance for the same hierarchical model and then compare decompositions
across hierarchical models because it is fair to compare posterior predictive variances
across models.

Expressions like \eqref{LTV2termgen} may be useful in a
practical sense as well because
predictive intervals (PI's) for $Y_{n+1}$ can be derived from the distribution of
\begin{eqnarray}
\frac{Y_{n+1} - E(Y_{n+1} \vert {\cal{D}}_n)}
{ \sqrt{Var(Y_{n+1} - E(Y_{n+1} \vert  {\cal{D}}_n) \vert {\cal{D}}_n)}  }
= \frac{Y_{n+1} - E(Y_{n+1} \vert {\cal{D}}_n)}
{\sqrt{\hbox{Var}(Y_{n+1} \vert {\cal{D}}_n) } }.
\label{PIform}
\end{eqnarray}
The denominator on the right in \eqref{PIform} is the posterior predictive variance and controls the length of the PI.
Our expressions for it allow us to identify the relative sizes of their
terms.
That is, because the posterior predictive variance ties multiple sources of variability together,
within a hierarchical model, we can look at relative contributions of terms to the PI.
For instance, it is meaningful to compare the sizes of terms such as
\begin{eqnarray}
\frac{E_{V_1 \vert {\cal{D}} } E_{V_2 \vert V_1,{\cal{D}}} [ \hbox{Var} ( Y_{n+1}\vert V_1, V_2, {\cal{D}}) ]}
{\hbox{Var}( Y_{n+1} \vert {\cal{D}})} 
\quad 
\hbox{and} 
\quad
\frac{E_{V_1 \vert {\cal{D}} } \hbox{Var}_{V_2 \vert V_1, {\cal{D}}} [E (Y_{n+1} \vert V_1, V_2, {\cal{D}})]}
{\hbox{Var}( Y_{n+1} \vert {\cal{D}})}.
\nonumber
\end{eqnarray}
A relative assessment of their contributions to the posterior predictive variance allows us to identify
the biggest contributions to the length of a PI.    Terms that do not contribute much, relatively, can be omitted
thereby identifying which terms are driving the width of PI's. We see an instance of this in an example in Sec. \ref{Draper}.

This decomposition is similar to \cite{Gustafson:Clarke:2004} who expanded the 
posterior variance $\hbox{Var}(\Theta \vert y^n)$.  However,  ours is predictive,
on a common scale, and hence
directly useful in expressions for PI's from, say,  \eqref{PIform}.
Moreover, in \cite{Gustafson:Clarke:2004}, the terms were
forced into a single `standard error' interpretation rather than treated
as distinct patterns of expectations and variances that could be interpreted
in the context of quantifying the variability in the levels of the hierarchy.

Another way these decompositions may be useful is in terms of reducing the number of levels
in the hierarchy.  Consider the last term in \eqref{LTV2termgen}.
There are two basic ways we can get
$\hbox{Var}_{V_1 \vert {\cal{D}}_n} (E(Y_{n+1} \vert V_1, {\cal{D}}_n)) = 0$.  First, the distribution of $V_1$
concentrates at a single value $V_1 = v_1$.    Second, the models i.e., values of $V_1$ that get
non-zero weights, give the same predictions given ${\cal{D}}$.  That is, 
\begin{eqnarray}
E(Y_{n+1} ; V_1=v_1, {\cal{D}}) = E(Y_{n+1} ; V_1=v_2 , {\cal{D}})
\label{intop}
\end{eqnarray}
for any $v_1$ and $v_2$ getting positive weight.
Identifyng these sets is essentially intractable.
However, by carefully selecting the models $V= v$ to ensure they
are different and having a large enough $n$ the chance of satisfying \eqref{intop} for
two values of $V_1$ will be vanishingly small. 
Thus, on pragmatic grounds, with some foresight,  if the last term on the right is chosen so it explicitly depends only
on a single component of $V$ and that term drops out i.e., is close to zero, we can simply 
set $V_1$ to be a constant
meaning that level of modeling drops out.  In a three term case
we would be left with only the first two terms on the right hand side that depend on $V_2$
in which $V_1$ was a constant.  The resulting expression reduces to \eqref{postpredvarV1}
but with $V_2$ in place of $V_1$..

The rest of this paper explores expressions such as \eqref{postpredvarV1} and \eqref{LTV2termgen}.
In Sec. \ref{Examples}, we present a variety of examples of instances of our conservation of posterior
predictive variance law.  In Sec.  \ref{decomposition} we present some generic results on
expressions in the LTV-scope of a posterior predictive variance..
In Sec.  \ref{Draper}, we revisit two examples from \cite{Draper:1995} and
by using an extension of his setting show how the relative sizes of terms in our decompositions can behave
in simulations.  
%In  Sec. \ref{Hamadieh}, we extend our simulations to compare how different
%decompositions (2 term and 3 term) behave. in a real example.  
In a concluding section, Sec. \ref{Discussion}, we
discuss the methodological implications of our representations for the posterior predictive variance.

\section{Generic Examples}
\label{Examples}

The reasoning in Sec. \ref{Intro} uses a generic hierarchical model.  There is no constraint on the levels in the
hierarchy except that they combine properly on one probability space.  That is, the whole inferential structure
satisfies the containment principle. of BAyesian statistics.   Even within the containment principle, the 
range of choices for $(K, V)$ is vast and two plausible hierarchical models may have very different behaviors.
In addition, as will be seen in Sec. \ref{Draper}, conditioning variables need not have any correlate in
reality; they may be aspects of modeling more commonly thought to be part of the likelihood.   So, each
hierarchical model and its LTV-scope must be examined individually before being compared with another.

The point of this section is to present a series of examples of hierarhical models that are amenable
to our LTV iterative procedure and to give a sense for their LTV-scopes.  We begin with the simplest cases
and move on to more complicated cases to develop intution for what the expressions in
the LTV-scope of a hierarchical model mean.

\subsection{A Two Level Hierarchical Model}
\label{2level}

The simplest hierarchical model has two levels i.e., has $K=1$:
\begin{eqnarray}
\Theta &\sim& w(\theta) \nonumber \\
Y &\sim& p(y \vert \theta),
\label{2levelBHM}
\end{eqnarray}
where $w$ is the density of a real parameter $\Theta = \theta$ and $p(\cdot \vert \theta)$ is the 
conditional density of $Y=y$, both with respect to Lebesgue measure.  
The posterior density is
\begin{eqnarray}
w(\theta \vert y^n) \propto w(\theta)p(y^n \vert \theta)
\nonumber
\end{eqnarray}
with normalizing constant 
\begin{eqnarray}
m(y^n) = \int w(\theta) p(y^n \vert \theta) {\rm d} \theta.
\end{eqnarray}
The posterior predictive density is now
\begin{eqnarray}
p(y_{n+1} \vert y^n) = \int p(y_{n+1} \vert \theta)w(\theta \vert y^n) {\rm d} \theta 
\nonumber
\end{eqnarray}
with mean
\begin{eqnarray}
E(Y_{n+1} \vert y^n) = \int y_{n+1} p(y_{n+1} \vert y^n) {\rm d} y_{n+1}
\nonumber
\end{eqnarray}
and
\begin{eqnarray}
\hbox{Var}(Y_{n+1} \vert y^n) = \int (y_{n+1} - E(Y_{n+1} \vert y^n))^2 p(y_{n+1} \vert y^n) {\rm d} y_{n+1} .
\nonumber
\end{eqnarray}
So, the LTV gives the posterior predictive variance as
\begin{eqnarray}
\hbox{Var}(Y_{n+1} \vert y^n) = E_\Theta ( \hbox{Var}(Y_{n+1} \vert \Theta, y^n) \vert y^n)+  \hbox{Var}_\Theta(E(Y_{n+1} \vert \Theta, y^n) \vert y^n) .
\label{LTVsimple}
\end{eqnarray}
The first term on the right is the variability of the high posterior probability predictive distributions.
The second term on the right is an assessment of how important the model used for prediction is.
This interpretation is, in fact, independent of the fact that $\Theta$ is a real parameter.   Only two-term
examples i.e., one usage of the LTV, admit a concise intepretation in general.   We will see that 
with two or more usages of the LTV and so $K \geq 2$, we get three or more terms and the interpretation
is much more complex and depends on the choice of conditioning variables.

In some cases,  \eqref{LTVsimple} can be worked out explicitly.  Let $Y_i \sim N(\theta, \sigma)$ be independent
and identically distributed (IID) for $i = 1, \ldots, n$ where $\theta \sim N(\theta_0, \tau^2)$ and $\theta_0$, $\sigma$ and $\tau$
are known.    It is easy to see that
$$
p(y_{n+1} \vert \theta, y^n) = \frac{p(y^{n+1}, \theta))}{p(y^n, \theta)} = p(y_{n+1} \vert \theta),
$$
where $\sigma$ has been suppressed in the notation.   So, it is also easy to see that
$$
E(Y_{n+1} \vert \theta, y^n) = \theta \quad \hbox{and} \quad \hbox{Var}(Y_{n+1} \vert \theta, y^n) = \sigma^2.
$$
Since $\hbox{Var}(Y_{n+1} \vert \theta, y^n)$ is a constant, its expectation under the posterior for $\theta$
is unchnaged.  Thus, the first term on the right in \eqref{LTVsimple} is
$$
E_{\Theta \vert y^n} \hbox{Var}(Y_{n+1} \vert \Theta, y^n) = \sigma^2.
$$
For the second term on the right in \eqref{LTVsimple} recall
the posterior for $\theta$ given $y^n$ is
\begin{eqnarray}
w(\theta \vert y^n) = \frac{1}{\sqrt{ 2 \pi \tau_n^2} } e^{-(1/2\tau_n^2) (\theta - \theta_n)^2}
\nonumber
\end{eqnarray}
where 
$$
\theta_n = \left( \frac{n}{\sigma^2} + \frac{1}{\tau^2} \right)^{-1} \frac{n}{\sigma^2}\left(\bar{y} + \frac{\theta_0}{\tau^2} \right)
\quad
\hbox{and}
\quad
\tau_n^2 = \left( \frac{n}{\sigma^2} + \frac{1}{\tau_0^2} \right)^{-1} .
$$
Now,
\begin{eqnarray}
\hbox{Var}_{\Theta \vert y^n}( E(Y_{n+1} \vert \Theta, y^n) ) = \hbox{Var}_{\Theta \vert y^n} (\Theta) =  \left( \frac{n}{\sigma^2} + \frac{1}{\tau_0^2} \right)^{-1} ,
\nonumber
\end{eqnarray}
and \eqref{LTVsimple} is
$$
\hbox{Var}(Y_{n+1} \vert y^n) = \sigma^2 + \left( \frac{n}{\sigma^2} + \frac{1}{\tau_0^2} \right)^{-1}  = \sigma^2 + {\cal{O}}(1/n).
$$
This reasoning extends to the case that $\sigma$ is assigned an inverse-Gamma distribution.

%\url{https://www.statlect.com/fundamentals-of-statistics/normal-distribution-Bayesian-estimation}

Another specific example that can be evaluated in closed form
sets $Y_i \sim  N(\mu, \sigma^2)$ to be
IID for $i=1, \ldots, n$ with the usual priors
$\mu \sim N(0, \sigma^2)$ and $\sigma^2 \sim InvGamma(\alpha, \beta)$ for $\alpha > 2$ and $\beta > 0$.
For contrast, note that a different problem with the Beta-Binomial gives that the `Var-E' term dominates; see
\cite{Casella:Berger:2002}, p. 168.

\cite{Draper:1995} argues that the first term on the right (a sequence of conditional expectations following  
a conditional variance) is the main assessment of variability that authors consider.    This is consistent with
the normnal example above and holds more generally.
Indeed, the first term on the right in \eqref{LTVsimple} is
\begin{eqnarray}
&& \int  \int \left((y_{n+1} - E(Y_{n+1} \vert \theta, y^n))^2p(y_{n+1} \vert \theta, y^n) \right) {\rm d} y_{n+1}w(\theta \vert y^n) {\rm d} \theta \nonumber\\
&=& \int  \int \left((y_{n+1} - E(Y_{n+1} \vert \theta))^2p(y_{n+1} \vert \theta) \right) {\rm d} y_{n+1}w(\theta \vert y^n) {\rm d} \theta \nonumber\\
&=& \int  \hbox{Var}_\theta (Y_{n+1}) w(\theta \vert y^n) {\rm d} \theta
\label{LTVsimpleterm1}
\end{eqnarray}
and
the second term on the right in \eqref{LTVsimple} is
\begin{eqnarray}
 && \hbox{Var}_\Theta \left( E(Y_{n+1} \vert \Theta) \vert y^n) \right)
\nonumber\\
&=& \int \left(  E(Y_{n+1} \vert \theta) -   \int E(Y_{n+1} \vert \theta) w(\theta \vert y^n) {\rm d} \theta \right)^2 w(\theta \vert y^n) {\rm d} \theta
\nonumber\\
&=& \int E^2(Y_{n+1} \vert \theta) w(\theta \vert y^n) {\rm d} \theta  - \left(  \int E(Y_{n+1} \vert \theta) w(\theta \vert y^n)
{\rm d} \theta \right)^2.
\label{LTVsimpleterm2}
\end{eqnarray}

When the posterior concentrates at a true value $\theta_0$, in distribution, $L^1$ or a.e., as $n \rightarrow \infty$, \eqref{LTVsimpleterm1} converges to $\hbox{Var}_{\theta_0} (Y_{n+1})$ and \eqref{LTVsimpleterm2}
converges to zero in the same mode.   So, the first term
asymptotically dominates.  This reasoning holds anytime the posterior concentrates as it typically does 
in ${\cal{M}}$-closed problems; more generally, see \cite{Berk:1966}. 
 However, this says little about the relative sizes of the two terms in 
finite samples.  

In the general case, the inner expressions on the right in \eqref{LTVsimple} are 
$\hbox{Var}(Y_{n+1} \vert \theta, y^n)$ and $E(Y_{n+1} \vert \theta, y^n)$ and
they have different meanings.  
In particular, the first term is small 
when $\hbox{Var}(Y_{n+1} \vert \theta, y^n)$ is small over the typical region of $\theta$ under the posterior
and the second term is small
when $E(Y_{n+1} \vert \theta, y^n)$, as a function
of $\theta$, changes little, again over the typical region of $\theta$.  
Loosely, the difference is whether the variance is small or the mean changes little.

If we know for sure that the mean changes little, i.e., $E(Y_{n+1} \vert \theta, y^n)$ is nearly
constant over the range of $\theta$'s most likely under the posterior, then
$$
\hbox{Var}(Y_{n+1} \vert y^n) \approx E_\Theta ( \hbox{Var}(Y_{n+1} \vert \Theta, y^n) \vert y^n).
$$
This latter expression is esentially what is observed in Sec. 4 of \cite{Dustin:Clarke:2023}.
However,  if we know for $y^n$ that the variance is small, i.e., for the $\theta$'s most likely under the posterior
we have that $\hbox{Var}(Y_{n+1} \vert \theta, y^n)$ is small, then
$$
\hbox{Var}(Y_{n+1} \vert y^n) \approx \hbox{Var}_\Theta(E(Y_{n+1} \vert \Theta, y^n) \vert y^n) .
$$

Another way to interpret \eqref{LTVsimple} is as follows.
When the `E-Var' term is large,  relative to the `Var-E' term,  there is more variability in 
the predictive distributions from the high posterior probability models than there is variability across models
so the model doesn't matter very much; all the commonly occurring models (high posterior probability)
are good.
When the `Var-E' term is large relative to the `E-Var' term it means that the specific model used for prediction
is much more important than the variability within models used for prediction. 

\subsection{Bayesian Model Averages I}
\label{BMA1}

One step up from \eqref{2levelBHM} we can consider a Bayesian model average (BMA).
Let $j= 1, \ldots, J$ index a collection of models ${\cal{M}} = \{ M_1, \ldots , M_J\}$.
Assume each $M_j$ consists of a likelihood $p(y \vert \theta_j)$ and a prior 
$w(\theta_j, j) = w(\theta_j \vert j )w(j)$ where the across models prior $w(j)$ is discrete.
Writing $J$ for $j$ as a random variabel as well as for the number of models will cause no confusion
because the context will indicate which is meant.  Now,
we can represent this as a two level hierarhical model
\begin{eqnarray}
(J, \theta_J) &\sim& w(\theta_j,j) \nonumber \\
Y &\sim& p(y \vert \theta_j).
\label{BMA2term}
\end{eqnarray}
Now, the $L^2$ BMA predictor is
\begin{eqnarray}
E(Y_{n+1} \vert y^n) = \sum_{j=1}^J E(Y_{n+1} \vert y^n, M_j)W(M_j \vert y^n).
\label{L2BMApred}
\end{eqnarray}
In \eqref{L2BMApred},  the two conditioning random variables, namely $J$ and $\theta_j$ are treated explicitly and implicitly, respectively.   In this case, it is not hard to see that one usage of the LTV recovers the usual formula for the posterior variance.

Indeed,  using the expression for posterior variance from p. 383 of \cite{Hoeting:etal:1999}, we find that
\eqref{L2BMApred} is
\begin{eqnarray}
\hbox{Var}(Y_{n+1} \vert y^n) 
 &=&\sum_{j=1}^J \hbox{Var}(Y_{n+1} \vert y^n, M_j) W(M_j \vert y^n)
\nonumber \\
&& \quad+ \sum_{j=1}^J E(Y_{n+1} \vert y^n, M_j)^2 W(M_j \vert y^n)
- E(Y_{n+1} \vert y^n)^2 \nonumber \\
&=&  E(\hbox{Var}(Y_{n+1} \vert M_J, y^n) \vert y^n) +
\hbox{Var} (E(Y_{n+1} \vert M_J, y^n)).
\label{LTVBMAsimple}
\end{eqnarray}
So \eqref{LTVBMAsimple} is the result of using the LTV and conditioning on $M_k$.
We have treated $\theta_j$ implicitly by integrated over it
before conditioning on the $M_j$'s.   Reversing this i.e., integrating over $j$ and using the LTV
with $\Theta_k$'s would have been inappropriate for BMA.  However, 
we shall see that different treatments of conditioning variables, when they make sense, typically give different
terms for the same $\hbox{Var}(Y_{n+1} \vert y^n)$.

Let us interpret \eqref{LTVBMAsimple} similarly to how we interpreted  \eqref{LTVsimple} but using
the $M_k$'s, not the $\theta_k$'s.
When the first term on the right `E-Var' is large,  we see most variability is in the predictive
distributions from the high posterior probability models rather than from the variability across models.
The second term on the right being small means that it doesn't matter very much which model 
you use for prediction.  On the other hand, if Var-E is large, model selection is important
but the smallness of the E Var term means the high posterior probability
models are good.

For the sake of completeness, let us record another two term expression in the scope of \eqref{BMA2term}
but conditioning on $(J, \Theta_J)$ as a two dimensional random variable:
\begin{eqnarray}
\hbox{Var}(Y_{n+1} \vert y^n) =  
\hbox{Var}_{\Theta_K,K}(E(Y_{n+1} \vert \Theta_K, K, y^n) ) 
+ E_{\Theta_K, K} ( \hbox{Var}(Y_{n+1} \vert \Theta_K, K, y^n) ).
\label{3rdwayBMA}
\end{eqnarray}
This is different from \eqref{LTVBMAsimple} where we mixed out over the $\Theta_j$'s before
examining the variability in $M_J$.  That is, in \eqref{LTVBMAsimple}, $\Theta_j$'s are implicit
whereas in \eqref{3rdwayBMA} they are explicit.

\subsection{Bayesian Model Averages II:  Three Level Hierarchical Model}
\label{BMA2}

Now write \eqref{BMA2term} as an equivalent three level hierarchical model
\begin{eqnarray}
J &\sim& w(j) \nonumber \\
\theta_j \vert J=j&\sim& w(\theta_j \vert j) \nonumber \\
Y &\sim& p(y \vert \theta_j).
\label{BMA2termequiv}
\end{eqnarray}
If we apply the law of total variance first to bring $M_j$ into $\hbox{Var}(Y_{n+1} \vert y^n)$
we get \eqref{LTVBMAsimple}. 
If we then use the LTV again in the first term on the right in \eqref{LTVBMAsimple}
to bring in $\theta_j$, we get
\begin{eqnarray}
\hbox{Var}(Y_{n+1} \vert y^n, M_j) &=& E_{\Theta_j \vert y^n, M_j} \hbox{Var}_{Y_{n+1} \vert y^n, M_j, \theta_j} (Y_{n+1} \vert  y^n, M_j,  \Theta_j = \theta_j , ) 
\nonumber \\
&& + \hbox{Var}_{\Theta_j \vert y^n, M_j} E_{Y_{n+1} \vert  y^n, M_j, \theta_j}(Y_{n+1} \vert y^n, M_j,  \Theta_j = \theta_j ).
\label{BMA3term}
\end{eqnarray}
Using \eqref{BMA3term} in \eqref{LTVBMAsimple} gives
\begin{eqnarray}
\hbox{Var}(Y_{n+1} \vert y^n) 
&=&  E_J E_{\Theta_J \vert y^n, M_J} \hbox{Var}_{Y_{n+1} \vert y^n, M_J \Theta_j} (Y_{n+1} \vert  y^n, M_J,  \Theta_J) 
\nonumber \\
&& + E_J \hbox{Var}_{\Theta_J \vert y^n, M_J} E_{Y_{n+1} \vert  y^n, M_J, \Theta_J}(Y_{n+1} \vert y^n, M_J,  \Theta_J ) \nonumber \\
&&+
\hbox{Var}_J (E(Y_{n+1} \vert M_J, y^n)),
\label{LTVBMA3term}
\end{eqnarray}
an instance of \eqref{LTV2termgen}.

In \eqref{LTVBMA3term} we conditioned first on $M_J$ and then on $\Theta_J$ because the $\Theta_j$'s is naturally nested 
in the $M_j$'s.  There is nothing to prevent us from setting up a mathematical structure in which we
can condition on $\Theta_j$ first and $M_j$ second but that is not the natural way to think about this situation.
We will see an example in Subsec. \ref{simbinom} where either order of conditioning makes sense.
In general, however,  it can be seen that the order of conditioning 
affects which terms appear but the value of $\hbox{Var}(Y_{n+1} \vert y^n)$ on the left is fixed once the
hierarchy is fixed.   In particular, the two models \eqref{BMA2term} and \eqref{BMA2termequiv} are the same so the
posterior variances on the left in \eqref{LTVBMAsimple} and \eqref{LTVBMA3term} are equal pointwise in $y^n$.  
Hence the the expressions on the right are equal albeit different and are in the scope of the model.
We can choose whichever sums of term in the scope of a given model we want depending on the
variabilities of m odeling quantities that concern us most. 

It is seen from \eqref{LTVBMA3term} that a three level hierarchy can
lead to a three term expression for the posterior predictive variance because
we have used the law of total variance twice, one for each level of the
hierarchy above the likelihood.  In general, each usage of the LTV generates one extra term.

We can calso apply the law of total variance to the second term in \eqref{LTVsimple},
i.e., the last term on the right in \eqref{LTVBMA3term}.
However, that will bring in the conditional expectation of a conditional variance of a
conditional expectation which will not simplify.  Such terms while mathematically correct 
are very difficult to handle.
Moreover, the terms in \eqref{LTVBMAsimple} treat the $\Theta$ as latent and 
so depend on its distribution even though it is not explicitly indicated.  
For this reason, here, we only apply the LTV to the variances that
occur in the leading term, i.e., the one of the form `$E$ $\hbox{Var}$', not any that
have a `$\hbox{Var}$ $E$'.  That is, while the full scope of a posterior predictive
variance contains many terms from using the LTV in all possible ways, we focus
on the subset of the scope where each term has exactly one variance operation
that moves from left to right with appropriate conditioning.   By analogy of such 
sums of terms with ANOVA decompositions, see \cite{Dustin:Clarke:2023}
we call this this the Cochran Scope
or C-Scope for short and henceforth limit our attention to sums of that form.

In this treatment of posterior variance the relative size of the terms is a tradeoff among
the size of model list,  the proximity of the parametric models on the list to each
other,  the across-models prior weights on models on the list, and the within-model priors.
It's no longer purely a probabilistic model.  We have to choose which terms we
want to control in our model selection.

Another hierarchical model to which our considerations apply is studied in Subsec. \ref{simbinom}.
It conditions an outcome on possible sets of selected variables and on three possible link functions in a generalized linear
model.    This is quite different from \eqref{BMA2termequiv} because the levels in the hierarchy 
include aspects of modeling typically associated with likelihoods not priors.   However,
our generic decompositions of posterior predictive variances accommodates these possibilities readily.

\section{Generic Decompositions for the Posterior Predictive Variance:  $C$-scope case}
\label{decomposition}

Recall the generic hierarchical model \eqref{hierarchicalmodelgeneric}.
Limiting attention to $V_1$, the LTV can be applied to give
\begin{eqnarray}
\hbox{Var}(Y_{n+1} \vert {\cal{D}}_n) = E ( \hbox{Var}(Y_{n+1} \vert V_1, {\cal{D}}_n) )
+ 
 \hbox{Var} (E (Y_{n+1} \vert V_1, {\cal{D}}_n)).
\label{lawtotalvariance1}
\end{eqnarray}
In \eqref{lawtotalvariance1}, $\hbox{Var}(Y_{n+1} \vert {\cal{D}}_n)$ looks the same as in other expressions
such as \eqref{LTVsimple} and \eqref{LTVBMA3term} but in fact it depends on the full hierarchy in
\eqref{hierarchicalmodelgeneric} because the posterior predictive variance ties all levels of
the hierarchy together.  That is levels 2 through $K$ in \eqref{hierarchicalmodelgeneric} affect
the posterior predictive variance on the left -- and the terms on the right -- even though they
are suppressed in the notation.

We can now apply the law  of total variance iteratively to itself,  i.e.., to the
first term -- `E Var' -- in \eqref{lawtotalvariance1} by introducing conditioning on $V_2$.
We can do the same in the new `E  E Var' term with $V_3$
and so on, generating one new term for each $V_k$ at each iteration.  
Overall, this gives us $K+1$ terms 
involving means and variances.   
The expression for $K=2$ is now
\begin{eqnarray}
Var(Y_{n+1} \vert  {\cal{D}}_n)&=  E_{V_1}E_{V_2}Var(Y_{n+1} \vert  {\cal{D}}_n, V_1, V_2) + E_{V_1}Var_{V_2}E(Y_{n+1}\vert  {\cal{D}}_n, V_1, V_2) \nonumber \\ 
   & + Var_{V_1}E(Y_{n+1} \vert  {\cal{D}}_n, V_1) 
\label{2unidims}
\end{eqnarray}
The
left hand is a constant (given ${\cal{D}}_n$) independent of the order of conditioning on the right
although different orders of conditioning will give different terms and of course, different hierarchical models will
have different psoterior variances.

\subsection{Overall Structure}

To generalize this and hence
quantify the uncertainty of the subjective choices we must make,
recall $V= (V_{1}, \ldots , V_{K})$, where $V_k$ 
represents the values of the $k$-th potential choice that must
be made to specify a predictor. 
Analogous to terminology
in ANOVA, we call $V_k$ a \textit{factor} in the prediction scheme, and 
we define the levels of $V_k$ to be $v_{k1}, \ldots , v_{km_k}$.  That is,  $v_{k\ell}$ is a 
specific value $V_k$ may assume. Thus, we take $V$ to be discrete having probability mass 
function $W(v)=W(V_1=v_1 \ldots, V_K = v_K)$.  
Effectively we are assuming that any continuous parameters are at the first
level of the hierarchy above the likleihood and have been integrated out as in
the BMA example in Subsec. \ref{BMA1}.
The $V_k$'s are not in general
independent and $W$ corresponds to a prior on $V$. 
Define our chosen model list by
 \begin{equation}
\label{pred_scheme}
{\mathcal{V}}^{K} =  \{v_{11} ,\ldots , v_{1m_1}\} \times \ldots \times \{v_{K1} ,\ldots , v_{Km_K}\}.
\nonumber
\end{equation}
There are $m_1 \times \cdots \times m_K$ distinct models in ${\cal{V}}^{K}$ 
and they may or may not have a hierarchical structure.  

Our first result gives a decomposition of the posterior predictive variance by
conditioning on $V$.
The general $K$ case is seen in Clause (i) of Prop. \ref{General_pred_variance_prop}.
However, the order of conditioning will give different terms on the right.
The ordering is chosen so that the most important terms to the analyst can be
readily assessed.   Permuting the $V_k$'s generates all the
decompositions in the $C$-scope of the posterior predcitive variance from the BHM.
Recall \eqref{LTVP} and \eqref{LTV2termgen}.  These are one and two conditioning variable
versions of decompositions in the C-scope of a posterior predictive variance.
Clause (ii) of Prop. \ref{General_pred_variance_prop} is a variant on Clause (i) 
from collapsing all the levels
in the hierarchy above the likelihood into a single conditioning variable.

\begin{proposition}
\label{General_pred_variance_prop}
We have the following two expressions for the posterior predictive variance when the 
factors correspond to a model list.\\
{\it Clause (i):}
For $K=1$ in \eqref{hierarchicalmodelgeneric} we have 
\begin{eqnarray}
\hbox{Var}(Y_{n+1} \vert {\cal{D}}_n) = E ( \hbox{Var}(Y_{n+1} \vert V, {\cal{D}}_n) )
+ 
 \hbox{Var} (E (Y_{n+1} \vert V, {\cal{D}}_n)).
\label{lawtotalvariance}
\end{eqnarray}
and, for
$K \geq 2$ in \eqref{hierarchicalmodelgeneric}, the posterior predictive variance of $Y_{n+1}$ as function of 
the factors defining the predictive scheme is 
\begin{align}
\label{Conditional_Var_sum}
\nonumber 
Var(Y_{n+1} \vert  {\cal{D}}_n)(  {\mathcal{V}}^{K}) 
& = E_{(V_1,\ldots, V_k)} Var(Y_{n+1} \vert  {\cal{D}}_n, V_1, \ldots, V_K) \\ \nonumber 
& + \sum_{k=2}^{K} E_{(V_{1}, \ldots, V_{k-1} )} Var_{V_k}E(Y_{n+1} \vert  {\cal{D}}_n, V_1, \ldots, V_k) \\
& + Var_{V_1}E(Y_{n+1} \vert  {\cal{D}}_n, V_1).
\end{align}
\\
{\it Clause (ii):}
For any $K$, the posterior predictive variance 
$Var(Y_{n+1} \vert  {\cal{D}}_n)(  {\mathcal{V}}^{K})$ can be condensed into a two term decomposition:
\begin{align}
\label{condensed_var}
\nonumber Var(Y_{n+1} \vert  {\cal{D}}_n)( {\mathcal{V}}^{K}) 
& = E_{(V_1,\ldots, V_K)} Var(Y_{n+1} \vert {\cal{D}}_n, V_1, \ldots, V_K) \\ 
& + Var_{(V_1,\ldots, V_K)}E(Y_{n+1} \vert  {\cal{D}}_n, V_1, \ldots, V_K).
\end{align}
\end{proposition}
\begin{proof}
The proof of {\it Clause i)}  is a straightforward iterated application of the law of total variance and
{\it Clause ii)}  follows from the law of total variance simply treating $V$ as a vector
rather than as the string of its components. %$\square$
\end{proof}

As the number of conditioning variables increases, conditional variances tend to decrease.  Here, however,
that intuition breaks down because we condition on all $V_k$'s.  The question is the effect of the 
relative placement of the expectations relative to the variances.  In a limiting sense, one expects
that with IID data expectations will converge to well-defined limits, even if they are
random variables due to the conditioning thereby giving nonzero variances.  In our work we have 
focused on the last term in \eqref{Conditional_Var_sum}, often finding it to be small.
This may be an artifact of the examples we have examined.  On the other hand, our
intution is that typically the leading term in \eqref{Conditional_Var_sum} should dominate,
see Subsec. \ref{2level},
because the posterior for $V_1, \ldots , V_K$ should concentrate at a fixed value,
namely the true value if it is reachable and the value closest to the true value
(in relative entropy sense) if it is not.

A key question is what the cardinality of the $C$-scope is.
Consider \eqref{LTV2termgen} or equivalently \eqref{Conditional_Var_sum} with $K=2$.
There are five cases.
The first is \eqref{2unidims} in which $V_1$ and $V_2$ are treated
as individual conditioning variables.    The second is the same but with $V_1$ and $V_2$
interchanged.  
The third is \eqref{lawtotalvariance1}
where $V_1$ is replaced by $(V_1, V_2)$, i.e., $(V_1, V_2)$ is regarded as a single
random variable that happens to be bivariate.  The fourth is 
\eqref{lawtotalvariance1} in which $V_2$ is regarded as integrated out before
conditioning.  In this case, we say $V_2$ is latent and $V_1$ is manifest.
The fifth case is the reverse, i.e.,  we treat $V_1$ as latent and $V_2$ as manifest.
Expression \eqref{LTVBMAsimple} is an example of this if $V_2=\Theta_K$ and
$V_1=K$.  
In all five cases, $\hbox{Var}(Y_{n+1} \vert y^n)$ is the same even though the decompositions are
different.   So, we have five different ways to model the posterior
variance in the same hierarchical model.  While they are equivalent
mathematically and come from the same hierarchical model they are not
equivalent statistically.  This counting argument can be generalized to arbitrary $K$.
(Write out new proposition giving $\#(C-\hbox{scope})$ as a function of $K$.)

\subsection{Choosing $V_K$}
\label{choosingVk}

In the usual physical modeling scenario the decompositions presetned here will be of limited
use largely because scientists don't often include irrelevant quantities in their models.
However, when the BHM is not physically motivated, e.g., the $V_k$'s are mathematical
aspects of the likelihood, it will often not be clear which $V_k$'s or values of $V_k$'s
are important to retain.  Reducing the number of values of a $V_k$ is not as important
as dropping a $V)k$, but as in the factor levels of ANOVA can be useful.

An example may help. 
One choice of $V_K$, with $K=2$, that can often be used to expand a model list
to capture more possibilities and then winnow down to the most successful of them
is the following.    Consider trying to assess the importance of sets of variables. 
Suppose we have a list of models
${\cal{M}} = \{m_1, \ldots, m_q \}$ and a set of explanatory variables 
${\cal{X}} = \{ X_1, \ldots, X_p\}$.   If $q=2$, we may have $m_1$ 
as the linear model and $m_2$ as some choice of non-linear model.
Write ${\cal{P(X)}} = \{ \{X\}_1,\ldots , \{X\}_{2^p} \}$ to mean
the power set of ${\cal{X}}$. Now we can consider each model with each subset of 
explanatory variables as inputs to the modeling.   Here, $V_1$ corresponds to
the uncertainty in the predictive problem due to the models and $V_2$
corresponds to the variable selections we use in the models.   We use a version
of this in Subsec. \ref{disaster}.  The idea is that the experimenter has little
information about which variables are included so the formal mathematical approach
is to include all the a priori plausible ones and examine the predictive 
properties of the resulting model.

Now, we can use the decomposition in Subsec. \ref{BMA1} or \ref{BMA2}.  In addition,
using a Bayes model average we write the posterior predictive density
\begin{eqnarray}
p(Y_{n+1} \vert {\cal{D}}_n) 
= \sum^q_{i=1} p( m_i \vert {\cal{D}}_n) \sum^{2^p}_{j=1} 
p(\{X\}_j \vert {\cal{D}}_n, m_i ) p(Y_{n+1} \vert {\cal{D}}_n, \{X\}_j, m_i),
\label{postpredBMA}
\end{eqnarray}
generically denoting densities as $p$.
Now, we can calculate the posterior probability for each 
set of explanatory variables from
\begin{eqnarray}
 p(\{X\}_j\vert {\cal{D}}_n) = \sum^q_{i=1} p( m_i \vert {\cal{D}}_n) p(\{X\}_j \vert {\cal{D}}_n, m_i ).
 \label{postprobvarset}
\nonumber
\end{eqnarray}
This posterior probability is a measure of ``variable set importance''.  A similar expression gives 
a measure of importance for an individual model.   
Thus, we anticipate that when the levels in the BHM represent mathematical quantities that
we do not have a physical basis for including or excluding we can use our 
decompositions to assess whether factors $V_k$ or factor levels of a $V_k$ are worth including.

We can represent any conditioning quantity 
as $V = (V_1, \ldots , V_K)^T$.   As in Subsec. \ref{2level} for $K=1$, $V_1$ might 
simply be a parameter.  As in Subsec. \ref{BMA2},
for $K=2$, $V_1$ might be a model and $V_2$ might correspond to a parameter.
Or, as in Sec. \ref{Draper}, $V_1$ may correspond to the choice of link function in a GLM
while $V_2$ may correspond to selections of explanatory
variables (as above).  
Thus,  we must choose
a $K$ and we can regard each $V_k$ as an aspect of a modeling strategy.  
For instance,  if $K=2$, $V_1$ may be a `scenario' and $V_2$ may be a `model' in the sense of
\cite{Draper:1995}, a parallel we develop in Sec. \ref{challenger}.  
We will write as if the $V_k$'s
are discrete modeling choices remembering that the law of total variance
applies for continuous random variables as well.

Note that the sort of hierarchical modeling we advocate here can be artificial in the
sense that we use mathematical modeling, simplify the model down to the quantities that seem to
matter predictively, and then use the resulting model to form PI's.  Once good predcition 
has been achieved the
quest for a more realistic  model (that probably will not perform as well predictively)
can begin and compared with the formal model our approach should yield.
Consequently, we advocate the generation of multiple BHM's, using different
mathematical features.  In Sec. \ref{Draper}, we use link functions in a GLM as a conditioning
variable.   In \cite{Dustin:Clarke:2023}, we used a single $V$ to represent
the choice of a shrinkage method in penalized linear regression. 

Typically, one of the most important levels in  a BHM is the selection over models.
We can enlarge model list simply by including more plausible models.
However, this may lead to problems such as dilution; see \cite{George:2010}.
So, we want to assess the effect of a model list on the variance
of predictions.    Consider a model list ${\cal{M}}$ and suppose
we don't believe it adequately captures the uncertainty (including mis-specification)
of the the predictive problem. 
We can expand the list by including other competing 
models and this can be done by adding more models to it or by embedding the
models on the list in various `scenarios' as is done in \cite{Draper:1995}. 
Once a new model 
list ${\cal{M}'}$ is constructed, if it contains new models with positive posterior probability, 
the posterior predictive distribution $p(Y_{n+1} \vert {\cal{D}}_n)$ resulting from
${\cal{M}'}$ will be differ from $p(Y_{n+1} \vert {\cal{D}}_n)$ resulting from
${\cal{M}}$.    Hence, we can use the decompositions here to help decide which
of model ${\cal{M}}$ and ${\cal{M}'}$ is more reasonable and hope that both
simplify to the same predictor.

In reality, the relative sizes of terms in the various
decompositions of the form \eqref{Conditional_Var_sum}
depend heavily on the 
choice of $K$, $V_K$, and the likelihood.   Fortunately,
in practice, we usually only have one largest $V_K$ that we most want to consider
even though different orderings of the
$V_k$'s will affect which terms appear in the decomposition.    One benefit of more elaborate
hierarchical models is that we can use the levels to specify a likelihood so
that features of the likeilihood can be examined via our decompositions.
Once $K$ and $V_K$ have been chosen, the posterior predictive 
variance decomposition based on $V_K$ can
be generated and examined for which terms are important.  
We regard the selection of 
$V$ in general as an aspect of mathematical modeling in that it seem hard to
propose a univeral method.  The appropriate choice of levels in the BHM will be data set dependent.
We have given examples here to show the flexibility of the approach.

Using stacking -- or any other model averaging procedure -- in place of the BMA
leads to results analogous to Prop. \ref{General_pred_variance_prop};
 see \cite{Dustin:Clarke:2023}.

\section{Draper's Work on Model Uncertainty}
\label{Draper}

 Here we redo and extend some of the key examples developed by Draper.

\subsection{Revisting Draper (1995)}
\label{challenger}

Here we apply our techniques to two examples given in \cite{Draper:1995} and one further
example that his second example motivates.
The first example involves predicting the price of oil; the second example
involves predicting the chance of failure of O-rings in a space shuttle at a new temperature.
Our third example for this section is an extension of the latter data type with
a more difficult variable selection problem.  
Draper's main point was when making predictions, we need to consider the uncertainty of the `structural' choices we make or we can be lead to bad decisions.    Here, we have
formalized Draper's concept of structural choices in our conditioning variable $V$.
One danger in poor structural choices is that a PI may be found that is
unrealistically small leading to over-confidence.

\subsubsection{Oil Prices}

In the oil prices example in \cite{Draper:1995} there are two structural components to
the modeling namely,  12 economic scenarios with 10 economic models nested inside them.
These components represent 120 models and hence
introduce model uncertainty that must be quantified to generate
good PI's.  

In Draper's analysis each model was used given the parameters of each scenario.
This corresponds to $K =2$ and a three term posterior predictive variance decomposition.
Let $s_i$ denote scenario $i$ and $m_{ij}$ be model $j$ within scenario $i$.
Write $s_i \in S$ and $m_{ij} \in M_i \subset M$ where $M_i$ is the set of models
for scenario $i$ and $M$ is the union of the $M_i$'s.   Now, we have
\begin{align}  
\nonumber Var(Y_{n+1} \vert  {\cal{D}}_n) (S,M)  &=E_S E_MVar(Y_{n+1} \vert  {\cal{D}}_n , S , M)  \\
\nonumber &+ E_S Var_M(E(Y_{n+1} \vert  {\cal{D}}_n , S, M) )\\
&+ Var_S(E(Y_{n+1} \vert  {\cal{D}}_n , S)).
\label{DraperOil}
 \end{align}
The corresponding decomposition given by Draper is
$\nonumber Var(Y_{n+1} \vert  {\cal{D}}_n) = 178 + 363 + 354 = 895$.  We
cannot recompute this example because neither the data nor the details on the scenarios or
models are available to us.   However, in this case it is seen that
the between-scenarios variance, i.e., term \eqref{DraperOil}, contributes about 40\% to
the posterior predictive variance.   The second term on the right,  the between-models
within scenarios variance, is also about 40\%  The variance attributable to the predictions
within models and scenarios is about 20\%.    (See Table \ref{vartable} for the definition of terms.)
Thus all the three terms must
be used when forming PI's.

\subsubsection{Challenger Disaster}
\label{disaster}

Making the decision
to launch the space shuttle at an ambient temperature at which the
various components had not been tested ended up being catastrophic
-- and could have been avoided had a proper uncertainty analysis had been done.  
Statistically, the error of the decision makers was to choose a single model from a
model list rather than incorporating all sources of predictive uncertainty into their analysis.
The goal of this example originally was to show that a correct analysis
of the various sources of uncertainty would have led to a PI for $p_{t=31}$,
the probability of
an O-ring failure (at $31^\circ$) of $(.33,1]$ i.e., too high for a launch to
be safe.
Our goal in re-analyzing Draper's example is to identify which sources of
uncertainty can be neglected.    

We have 23 observations of the number of damaged O-rings ranging from zero to
six (because each shuttle had six O-rings).    Each observation also has a temperature $t$
and a `leak-check' pressure $s$.   Following Draper's analysis we also use $t^2$
as an explanatory variable.  Thus we have 24 vectors, each of length four.

We assume the number of damaged O-rings follows a $Binomial(6, p)$ distribution
where $p$ is a function of the explanatory variables via one of three link functions,
logit, $c\log\log$, and probit.  Thus, we have structural uncertainty in the choice
of variables and in the choice of link function.  In our notation, we set
$V_1 = \{L, C, P \}$ for the choice of
link function, logit, $c\log\log$, and probit respectively. Also let
$V_2 = \{t, t^2, s, \text{no effect}\}$ where no effect means an intercept-only model.  
The 24 models are summarized in \ref{Tab_challenger2}.

\begin{table}[ht]
\caption{\textbf{List of models for the Challenger disaster data:}  This table lists
all 24 models under consideration broken down by their structural choices -- link functions
and explanatory variables.}
\label{Tab_challenger2}
\centering
 \resizebox{\textwidth}{!}{  
\begin{tabular}{ccccccccccccccccccc}
  \hline
 ${\cal{V}}^{(2)}$ &  $m_1$& $m_2$ & $m_3$ & $m_4$ & $m_5$ & $m_6$ & $m_7$& $m_8$ & $m_{9}$ & $m_{10}$ & $m_{11}$  & $m_{12}$ & $m_{13}$\\
  \hline
$V_1$   &  L & L & L & L & L  & L & L&L   &  C & C & C & C & C    \\
$V_2$  & $t$ &  $t^2$& $s$ & $t, t^2$ & $t, s$ & $t^2, s$ & $t, t^2, s$   & no effect  & $t$ &  $t^2$& $s$ & $t, t^2$ & $t, s$ \\
   \hline
%\end{tabular}
%\label{Tab_challenger}
%\end{table}

%\begin{table}[ht]
%\caption{dsf}
%\centering
%\begin{tabular}{ccccccccccccccccccc}
  \hline
 ${\cal{V}}^{(2)}$  & $m_{14}$ & $m_{15}$ & $m_{16}$& $m_{17}$& $m_{18}$ & $m_{19}$ & $m_{20}$ & $m_{21}$ & $m_{22}$& $m_{23}$& $m_{24}$  \\
  \hline
$V_1$  & C&C & C  &  P & P & P & P & P&P  & P  &P   \\
$V_2$  & $t^2, s$ & $t, t^2, s$   & no effect & $t$ &  $t^2$& $s$ & $t, t^2$ & $t, s$& $t^2, s$  & $t, t^2, s$   & no effect  \\
   \hline
\end{tabular}
}
\end{table}

In fact, Draper did not consider all of these models.   Essentially he put zero
prior probability on all models except for $m_1, m_4, m_5,m_7,m_8$, and $m_{15}$.
Accordingly,  he only considered the set
$$
 {\cal{M}} = \{m_1, m_4, m_5, m_7, m_8, m_{15} \}
$$
with a uniform prior.
Draper then gave a table of posterior quantities for the structural choices, and a
posterior predictive variance decomposition for within-structure and between-structure
variances as
\begin{eqnarray}
Var(p_{t=31}\vert  {\cal{D}}_{23})= Var_{within} + Var_{between} = 0.0338 + 0.0135 = 0.0473.
\label{twotermDraper}
\end{eqnarray}
That is, even though there were two structural choices, Draper used a decomposition
appropriate for one.  This corresponds to using our result Clause (ii) in
Prop. \ref{General_pred_variance_prop}.  Draper's conclusion was that
$.0135/.0473 \approx 28.5\%$ so the uncertainty represented by the second term
in \eqref{twotermDraper} could not be neglected.

Here we extend Draper's analysis and confirm that structural uncertainty should not be ignored.
For our implementation, we use the full set of 24 models but do not employ the same
approximations.   Then, we use the {\sf{BMA}} package in R to get
the posterior distributions of the parameters of the models and the posterior weights
for $V_2$.  We also use the ${\sf rjmcmc}$ package to get the posterior weights
for $V_1$.
The resulting posterior distributions are qualitatively similar to
Draper's approximate posteriors.

Considering all sources of uncertainty yields a posterior predictive variance decomposition of
\begin{align}
\label{var_challenger_3term}
\nonumber Var(p_{t=31}\vert  {\cal{D}}_{23})&=  E_{V_1}E_{V_2}Var(p_{t=31} \vert  {\cal{D}}_{23}, V_1, V_2) + E_{V_1}Var_{V_2}E(p_{t=31} \vert  {\cal{D}}_{23}, V_1, V_2) \\ \nonumber
   & + Var_{V_1}E(p_{t=31} \vert  {\cal{D}}_{23}, V_1) \\ \nonumber
   &=  0.01469 + 0.0996 + 0.0017 \\
   &= 0.11599.
\end{align}
This is almost three times the variance as obtained by Draper.  We confirm his intuition that
structural uncertainty was much greater than assumed when making the decision
to launch the shuttle.   Moreover, Draper commented that other analyses could lead to larger
posterior variances.  So, \eqref{var_challenger_3term} is consistent with his intuition.

Looking at the numbers in \eqref{var_challenger_3term} we can see the last is an order of 
magnitude smaller than the other two.
Thus, we conclude that the terms representing the between-models within-link functions variance
and the between-predictions within-models and links variance are terms that must be retained.
A frequentist testing approach confirms this; see \cite{Dustin:Clarke:2022}.
So, we would be led to consider a new hierarchical model that did not include $V_1$
but had a two term decomposition with a new value of
$\hbox{Var}(Y_{n+1} \vert {\cal{D}}_n)$ and we would compare it with the first two
terms on the right in \eqref{var_challenger_3term}
to see which expression for the posterior predcitive variance was more convincing.

\subsection{Simulated binomial example}
\label{simbinom}

Finally, we study a simulated example following the same structure as the Challenger data.
That is we simulated $n$ observations from a binomial generalized linear model
$$
Y_i \vert p_i \sim Binomial(30,p_i),
$$
where $p_i = \frac{1}{1+e^{-X'_i\beta}}$, in which $X'_i \sim N(0,1)$ is a $1 \times 10$ 
vector of explanatory variables, and 
$$
\beta = (0.75, 0.25, -0.3,0.5,0,0,0,0,0,0)
$$ 
is a 
$10 \times 1$ vector of true regression parameters. Here we let 6 entries in $\beta$ be zero to represent a meaningful model selection problem. In this problem, we again recognize 
three sources of structural uncertainty: predictive uncertainty across-models and link functions 
(`predictions'),
models across link functions (`models', $V_2$), and link functions (`links', $V_1$).  Our goal is to 
study the effect of the sample size on each each term in the posterior 
predictive variance decomposition,  as well as each of the test statistics.

We continue to use a three term decomposition like that in Subsec.  \ref{disaster}:
\begin{eqnarray}
\nonumber Var(Y_{n+1}\vert  {\cal{D}}_{n}) &=&  E_{V_1}E_{V_2}Var(Y_{n+1} \vert  {\cal{D}}_{n}, V_1, V_2)
+ E_{V_1}Var_{V_2}E(Y_{n+1} \vert  {\cal{D}}_{n}, V_1, V_2)
\nonumber \\
 &+&  + Var_{V_1}E(Y_{n+1} \vert  {\cal{D}}_{n}, V_1) 
\label{simbinomdecomp}
\end{eqnarray}
but our `$Y_{n+1}$' here is the number of successes in 30 trials, a random variable, as opposed to
a probability such as $p_{t=31}$.    In order, the three terms are predictions, models, and links.

As sample size increases,
the link functions proportionately contribute more and more
to the overall variance where as the models contribute less and the 
predictions are stable.  We comment that the overall variance actually
decreases with sample size so the relative importance of, say, links, may
increase even as its absolute importance decreases.

These conclusions are reinforced by Fig. \ref{simulated_binom}.
On the left panel we see that all four variances decrease with $n$, the
top curve representing the sum of the three lower curves.
The right panel shows that as expected the relative contribution of models
decreases monotonically.  It also shows that as $n$ increases, the curves for
links and predictions approach each other.  In simulation results not
shown here, the two curves actually cross around $n=475$ and suggest that
by $n=900$ or so that the curve for predictions will indicate a relatively small
contribution to the decreasing total variance curve compared to the relative 
contribution of links.  However, by this point, the total variance is so small
that the relative contributions of the terms is  not important.
\begin{figure}
\begin{center}
\begin{tabular}{cc}
\includegraphics[width=.45\columnwidth]{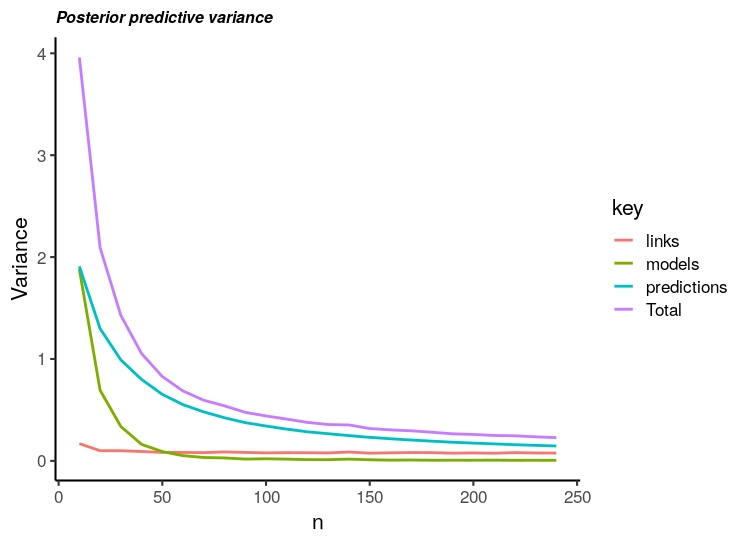}%}%
\includegraphics[width=.45\columnwidth]{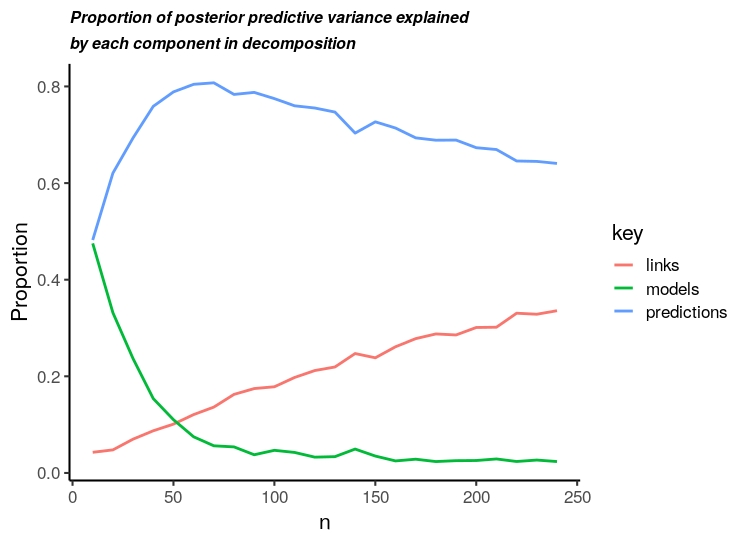}%}  
\end{tabular}
\end{center}
\caption{\textbf{Plots of posterior predictive variances as a function of $n$
 for each term in the decomposition. } Left: The actual values of the terms in the total posterior
predictive variance.  Right: The proportions each term contributes to the total variance.
 Both sets of curves show the results for link function, models, and predictions.}
\label{simulated_binom}
\end{figure}

Overall interpretation:  Models Term:  The models term decreases fast in both figures.  We think this means that
variable selection is fast so its contribution to the posterior variance is small.
Links term:  The term for links on the left stabilizes quickly so as the other terms decrease
absolutely it increases relatively.  Predictions term:  This decreases on the left but more slowly or
from a higher start that the other curves.  Then it decreases faster than the links term.
So because of the rates, it increases and decreases relatively in the right.

%\section{Real Data Analysis}
%\label{Hamadieh}

\section{Discussion}
\label{discuss}

The main point of this paper is to provide a decomposition of posterior predictive
variances for a future
value from a
Bayesian hierarchical model.  This is important because the posterior predictive variance
controls the width of prediction intervals and we want to know what aspects of variance
are contributing most to the width.  Our decompositions use a fixed BHM and hence
decompose a fixed posterior variance into terms using the law of total variability iteratively.
There are many possible decompositions for a posterior variance; they depend on the
ordering of the conditioning features in the BHM and whether features are ggrouped
together or not.  We focus on what we call the $C$-scope of a BHM.  This is the
collection of decompositions of the posterior predictive variance that arise from 
using the law of total variance
only on terms in which an expectation of a variance appears.

In examples, we have shown how some of these decompositions behave.  The idea is
to propose a BHM with levels that are mathematical and perhaps
 essentially artificial but nevertheless gives good predictions.  Then our approach can
winnow down the conditioning varaibles to simplify the model and identify which
terms contribute most to the width of PI's.

One drawback of the method is that it is unclear how to assess the relative contributions
of terms in the decomposition.  This is so because in general we do not have a likelihood
for these terms and therefore cannot do Bayes testing directly.  OTOH, there are
techniques that avoid the specification of likelihoods and priors relying instead on
a loss function structure to produce a viable posterior.  We have not investigated this
possibility, but it is promising as it is in the spirit of the mathematical modeling we
advocate here, namely being willing to use mathematical qusantities without
physical motivation as a way to produce predictive analyses.

Implicit in our work is the view that we should construct BHM not by identifying
levels with priors but rather by conditioning.  We can arbitrarily write condition expectations
and variance without concerning our selves with prior selection.    We will have to assign priors
at some point -- often simply uniform since most of the levels in the BHM will be discrete.
That is, rather than thinking our way physically to what a given $V_k$ should mean
we can proceed artificially choosing $V_k$'s to condition on simply by letting them correspond 
to mathematical features of models such as variable selection, choice of prior, choice of
decay parameter in shrinkage models, etc. etc.  By ignoring physical modeling we are
free to contruct any class of predictors we please and assess them.  We call this conditional
modeling as opposed to hierarchical modeling because it is a change in
prsepctive even if the resulting models are mathematically equivalent.
One obvious benefit of conditional modeling is that the
components of modeling correspond directly to the components of the variance decomposition.
We would expect that the `right' varaince decomposition would use the level
of the BHM in the same order as they appear in the BHM.  That is, we would imagine
conditioning first on the top level, second on the level below it and so on.

One of the unexpected effects of our approach is that by using variances
we are focusing on the metric properties of a model, not just its
probabilistic properties.  After all, hierarchical models are strictly probabilistic
whereas our assessment of variances is based on size.  
Thus, we are converting a probabilistic modeling strategy into a metric modeling strategy.

We conclude with two entertaining observations.
First, for $K=1$,
$V$ may represent
the choice of a shrinkage method i.e., a penalty term in penalized linear regression.  The penalty
corresponds to a prior, so our method includes a technique for assessing the 
variability due to prior selection, i.e., it can directly assess how much a level in the
hierarchy that correspond to a prior contributes.  As such it can be used to
assess prior selection provided there is a hyperparameter above it in the hierarchy.
In a sense, we can asssign a sort of standard error to prior selection in terms of
how much wide it makes a PI.

Second, the treatment we have given for variance can, in principle, be extended
to higher level moments even though it looks hard.  For instance,
\cite{Brillinger:1969} gives a way to calculate cumulants of a distribution that can be a posterior quantity.
He gives a formula similar to our Prop.  \ref{General_pred_variance_prop} and gives examples
using this result for sums of variables and mixture distributions.  The order of the
cumulants is arbitrary but lower orders would likely be easier to use.

\begin{acks}[Acknowledgments]
Dustin acknowledges funding from the University of Nebraska
Program of Excellence in Computational Science.  Both authors acknowledge computational support from the Holland Computing Center at the University of Nebraska.
\end{acks}

\bibliographystyle{ba}
\bibliography{references}

\begin{thebibliography}{9}
\newcommand{\enquote}[1]{``#1''}
\expandafter\ifx\csname natexlab\endcsname\relax\def\natexlab#1{#1}\fi
\expandafter\ifx\csname url\endcsname\relax
  \def\url#1{{\tt #1}}\fi
\expandafter\ifx\csname urlprefix\endcsname\relax\def\urlprefix{URL }\fi
\ifx\endbibitem\undefined \let\endbibitem\relax\fi

\bibitem[{Berk(1966)}]{Berk:1966}
Berk, R. (1966).
\newblock \enquote{Limiting Behavior of Posterior Distributions when the Model
  is Incorrect.}
\newblock {\em An. Math. Stat.\/}, 37: 51--58.
\endbibitem

\bibitem[{Brillinger(1969)}]{Brillinger:1969}
Brillinger, D. (1969).
\newblock \enquote{The calculation of cumulants by conditioning.}
\newblock {\em Ann. Inst. Math. Stat.\/}, 21: 215--218.
\endbibitem

\bibitem[{Casella and Berger(2002)}]{Casella:Berger:2002}
Casella, G. and Berger, R. (2002).
\newblock {\em {Statsitical Inference} 2nd Edition\/}.
\newblock Duxbury, Australia.
\endbibitem

\bibitem[{Draper(1995)}]{Draper:1995}
Draper, D. (1995).
\newblock \enquote{{Assessment and Propagation of Model Uncertainty}.}
\newblock {\em J. R. S. S. B\/}, 57(1): 45--97.
\endbibitem

\bibitem[{Dustin and Clarke(2022)}]{Dustin:Clarke:2022}
Dustin, D. and Clarke, B. (2022).
\newblock \enquote{Testing for the Important Components of Posterior Predictive
  Variance.}
\newblock {\em arXiv:2209.00636\/}.
\endbibitem

\bibitem[{Dustin and Clarke(2023)}]{Dustin:Clarke:2023}
--- (2023).
\newblock \enquote{Finding the Important Components of Predictive Variance by
  P-ANOVA.}
\newblock {\em Submitted to Stat. Anal. Data Mining\/}.
\endbibitem

\bibitem[{George(2010)}]{George:2010}
George, E. (2010).
\newblock \enquote{{Dilution priors: Compensating for model space redundancy}.}
\newblock In {\em IMS Collections Vol. 6\/}, 158--165. Inst. Math. Statist.
\endbibitem

\bibitem[{Gustafson and Clarke(2004)}]{Gustafson:Clarke:2004}
Gustafson, P. and Clarke, B. (2004).
\newblock \enquote{Decomposing Posterior Variance.}
\newblock {\em J. Stat. Planning and Inference\/}, 119: 311--327.
\endbibitem

\bibitem[{Hoeting et~al.(1999)Hoeting, Madigan, Raftery, and
  Volinsky}]{Hoeting:etal:1999}
Hoeting, J., Madigan, D., Raftery, A., and Volinsky, C. (1999).
\newblock \enquote{{Bayesian Model Averaging: A Tutorial}.}
\newblock {\em Statist. Sci.\/}, 14(4): 382--417.
\endbibitem

\end{thebibliography}

\end{document}